\documentclass[11pt]{article}

\usepackage{fullpage}
\usepackage{amsmath,amsthm,amsfonts,dsfont}
\usepackage{amssymb,latexsym,graphicx}
\usepackage{palatino}
\usepackage{mathpazo}
\usepackage{stmaryrd}
\usepackage{mathtools}
\usepackage{hyperref}
\usepackage[lined,boxed]{algorithm2e}
\usepackage{subfigure}
\usepackage{boxedminipage}
\usepackage{authblk}
\usepackage{caption}

\usepackage{geometry}
\newtheorem{theorem}{Theorem}[section]

\newtheorem{proposition}[theorem]{Proposition}
\newtheorem{lemma}[theorem]{Lemma}

\newtheorem{corollary}[theorem]{Corollary}
\newtheorem{definition}[theorem]{Definition}

\newcommand{\R}{\ensuremath{\mathbb{R}}}
\newcommand{\Z}{\ensuremath{\mathbb{Z}}}

 \newcommand{\eps}{\varepsilon} 
\renewcommand{\epsilon}{\varepsilon}

\renewcommand{\vec}[1]{\ensuremath{\mathbf{#1}}}

\newcommand{\basis}{\ensuremath{\mathbf{B}}}
\newcommand{\problem}[1]{\mathrm{#1}}

\newcommand{\poly}{\mathrm{poly}}

\newcommand{\scarequotes}[1]{``#1''}

\makeatletter
\def\imod#1{\allowbreak\mkern8mu({\operator@font mod}\,\,#1)}
\makeatother

\newcommand{\lat}{\mathcal{L}}
\newcommand{\gs}[1]{\ensuremath{\widetilde{#1}}}
\DeclareMathOperator{\dist}{dist}
\DeclareMathOperator{\spn}{span}
\DeclarePairedDelimiter\inner{\langle}{\rangle}
\DeclarePairedDelimiter\set{\{}{\}}
\DeclarePairedDelimiter\floor{\lfloor}{\rfloor}
\DeclarePairedDelimiter\ceil{\lceil}{\rceil}
\DeclarePairedDelimiter\length{\lVert}{\rVert}

\begin{document}

\title{Search-to-Decision Reductions for Lattice Problems with Approximation Factors (Slightly) Greater Than One}
\author{Noah Stephens-Davidowitz\thanks{Supported by the National Science Foundation (NSF) under Grant No.~CCF-1320188. Any opinions, findings, and conclusions or recommendations expressed in this material are those of the authors and do not necessarily reflect the views of the NSF.}}
\affil{Courant Institute of Mathematical Sciences,\\ New York
 University.\\
 \texttt{noahsd@gmail.com}}
\date{}
\maketitle

\begin{abstract}
We show the first dimension-preserving search-to-decision reductions for approximate SVP and CVP. In particular, for any $\gamma \leq 1 + O(\log n/n)$, we obtain an efficient dimension-preserving reduction from $\gamma^{O(n/\log n)}$-SVP to $\gamma$-GapSVP and an efficient dimension-preserving reduction from $\gamma^{O(n)}$-CVP to $\gamma$-GapCVP. These results generalize the known equivalences of the search and decision versions of these problems in the exact case when $\gamma = 1$. For SVP, we actually obtain something slightly stronger than a search-to-decision reduction---we reduce $\gamma^{O(n/\log n)}$-SVP to $\gamma$-unique SVP, a potentially easier problem than $\gamma$-GapSVP. 
\end{abstract}

\section*{Disclaimer}

When I wrote this paper, I was unaware of prior work due to Cheng~\cite{Cheng13} and Hu and Pan~\cite{HP2014}. Cheng gives an efficient dimension-preserving search-to-decision reduction for approximate SVP that is essentially identical to the deterministic SVP reduction that I present in Theorem~\ref{thm:SVPdet}. Hu and Pan show how to extend this to approximate CVP, which I also do in Theorem~\ref{thm:CVPdet}. In particular, my claim that the reductions presented in this paper were the first such reductions is clearly false.

I thank  Priyanka Mukhopadhyay and Divesh Aggarwal for bringing this to my attention, and I apologize to Cheng, Hu, and Pan for failing to properly cite them originally.

The main results of this paper (i.e., Theorems~\ref{thm:SVPintro} and~\ref{thm:CVPintro}) are still novel and achieve better parameters than the similar reductions in~\cite{Cheng13,HP2014}. I leave the rest of this paper as it was before I learned about~\cite{Cheng13,HP2014}.

\section{Introduction}

A lattice $\lat = \{ \sum a_i \vec{b}_i \, : \, a_i \in \Z\} \subset \R^n$ is the set of all integer linear combinations of linearly independent basis vectors $\vec{b}_1,\ldots, \vec{b}_n \in \R^n$. 

The two most important computational problems on lattices are the Shortest Vector Problem (SVP) and the Closest Vector Problem (CVP). For any approximation factor $\gamma = \gamma(n) \geq 1$, $\gamma$-SVP is the search problem that takes as input a lattice and asks us to find a non-zero vector in this lattice whose length is within a factor of $\gamma$ of the minimal possible value. $\gamma$-CVP is the search problem that takes as input both a lattice and a target vector $\vec{t} \in \R^n$ and asks us to find a vector in $\lat$ whose distance to $\vec{t}$ is within a factor of $\gamma$ of the minimal distance. The natural decisional variants of these problems are called GapSVP and GapCVP respectively. Specifically, $\gamma$-GapSVP asks us to approximate the length of the shortest non-zero vector of a lattice up to a factor of $\gamma$, and $\gamma$-GapCVP asks us to approximate the distance from $\vec{t}$ to the lattice up to a factor of $\gamma$.

All four of these problems are interesting for a wide range of approximation factors $\gamma$. Indeed, algorithms for these problems have found a remarkable number of applications in computer science (e.g.,~\cite{LLL82,Len83,Kan87,Odl90,JS98,NS01,DPV11}). And, over the past twenty years, many strong cryptographic primitives have been constructed with their security based on the (worst-case) hardness of $\gamma$-GapSVP with approximation factors $\gamma = \poly(n)$ that are polynomial in the dimension (e.g.,~\cite{Ajt96,MR07,GPV08, Gen09, Peikert09, Reg09,LPR10,BV11,BV14}).

Due to their importance, there has been much work towards understanding the relationship between these problems (and their many close relatives). Since the fastest known algorithms for these problems run in time that is exponential in the dimension $n$, even with $\gamma = \poly(n)$,
\emph{dimension-preserving} reductions between lattice problems are of particular importance~\cite{Kan87,GMSS99,Micciancio08,LM09,latticereductions,DGStoSVP}. 
Perhaps the best-known such reduction is the efficient dimension-preserving reduction from $\gamma$-SVP to $\gamma$-CVP (and from $\gamma$-GapSVP to $\gamma$-GapCVP) due to Goldreich, Micciancio, Safra, and Seifert~\cite{GMSS99}. 
This proves that the time complexity of $\gamma$-SVP, as a function of the dimension $n$, cannot be more than a polynomial factor higher than the time complexity of $\gamma$-CVP. We stress that we could \emph{not} reach this conclusion if the reduction increased the dimension significantly, which is why dimension-preserving reductions interest us.

As a much simpler example, we note that there is a trivial dimension-preserving reduction from $\gamma$-GapSVP to $\gamma$-SVP that works by just finding a short vector in the input lattice and outputting its length. There is of course a similar reduction for CVP as well. More interestingly, there are relatively simple dimension-preserving \emph{search-to-decision} reductions in the special case when $\gamma = 1$---i.e., finding \emph{exact} shortest vectors is no harder than computing the \emph{exact} lengths of shortest vectors, and finding exact closest vectors to targets is no harder than computing the exact distances between targets and lattices. (See, e.g.,~\cite{Kan87} or~\cite{MicciancioBook}, or simply consider the reductions in the sequel with $\gamma = 1$.) However, prior to this work, there were no known search-to-decision reductions for either SVP or CVP for any approximation factor $\gamma > 1$.

This state of affairs was quite frustrating because, with very few exceptions, our best algorithms for the decision problems work by just solving the corresponding search problem. In other words, we don't really know how to \scarequotes{recognize} that a lattice has a short non-zero vector (or a vector close to some target) without just finding such a 
vector.\footnote{\label{foot:searchdecisiongap}The author knows of three rather specific exceptions. There is an efficient algorithm for $\sqrt{n/\log n}$-GapCVP \emph{with preprocessing}~\cite{AharonovR04}, while the best efficient algorithm for search CVP with preprocessing only achieves factor of $\gamma = n/\sqrt{\log n}$~\cite{cvpp}. There is a $2^{n/2+o(n)}$-time algorithm for $2$-GapSVP for which no analogous search algorithm is known~\cite{ADRS15}. And, in the special case of ideal lattices in the ring of integers of a number field, $\gamma$-GapSVP is trivial for some values of $\gamma$ for which $\gamma$-SVP appears to be hard. (See, e.g.,~\cite{PR07}.)}
If there are better techniques, then we would be thrilled to find them!  But, if this extremely natural approach is actually optimal, then it would be nice to prove it by formally reducing the search problems to their decision variants. (Of course, it is conceivable that the search and decision problems have the same complexity, even if no search-to-decision reduction exists. One might reasonably argue that this is even the most likely scenario. But, we can at least hope that Nature would not be so unprincipled.)

The ideal positive result in this area would be an efficient dimension-preserving reduction from $\gamma$-SVP to $\gamma$-GapSVP for all $\gamma \geq 1$, and likewise for CVP. But, this seems completely out of reach at the moment (perhaps because no such reductions exist). So, as a more approachable goal, we can try to find non-trivial reductions that lose in the approximation factor. Indeed, as we mentioned above, we know that search and decision problems are equivalent in the exact case. Can it truly be the case that equivalence holds when $\gamma = 1$, but \emph{nothing} non-trivial holds for any $\gamma > 1$---even, say, a reduction from $n^{100}$-CVP to $(1+2^{-n})$-GapCVP?!

\subsection{Our results}

We make some progress towards resolving these issues by presenting dimension-preserving search-to-decision reductions for both approximate SVP and approximate CVP. Our reductions generalize the known equivalences in the exact case. But, they lose quite a bit in the approximation factor, and their running times depend on the decision approximation factor. They are therefore primarily interesting when the decision approximation factor is very close to one, as we explain below. 

\begin{theorem}[SVP reduction]
\label{thm:SVPintro}
For any $\gamma = \gamma(n) \geq 1$ and $a = a(n) \geq \log(n+1)$, there is a dimension-preserving (randomized) reduction from $\gamma^{n/a}$-SVP to $\gamma$-GapSVP that runs in time $2^{O(a)} \cdot \gamma^{O(n)}$.
\end{theorem}
Theorem~\ref{thm:SVPintro} is primarily interesting for any $\gamma \leq 1 + O(\log n/n)$ and $a = \Theta(\log n)$. For such parameters, the running time is $\poly(n)$ and the search approximation factor is $\gamma^{O(n/\log n)} \leq O(1)$. 
However, we note that the theorem is non-trivial whenever we have $1 < \gamma \leq 1+ \eps$ and $a \leq \eps n$, where $\eps >0 $ is some small universal constant.\footnote{In particular, we can choose $\eps$ so that, with $a = \eps n$ and $\gamma \leq 1 + \eps$, we get a reduction from $\gamma^{1/\eps}$-SVP to $\gamma$-GapSVP that runs in time $O(2^{n})$. For larger values of $a$ or $\gamma$, the reduction is subsumed by the known $2^{n+o(n)}$-time algorithm for SVP~\cite{ ADRS15}.} 

We actually reduce $\gamma^{n/a}$-SVP to $\gamma$-unique SVP, which is a potentially easier problem than $\gamma$-GapSVP. (See Definition~\ref{def:uSVP} for the formal definition of $\gamma$-unique SVP, and Theorem~\ref{thm:SVP} for the reduction.) The reduction described above then follows from this result together with Lyubashevsky and Micciancio's reduction from $\gamma$-unique SVP to $\gamma$-GapSVP~\cite{LM09}. We obtain a few additional corollaries as well. E.g., this shows a dimension-preserving reduction from $\sqrt{n}$-CVP to $\gamma$-unique SVP (and thus to $\gamma$-GapSVP as well) that runs in time $\poly(n) \cdot \gamma^{O(n)}$. This also gives an alternative and arguably more natural proof of Aggarwal and Dubey's result that $\gamma$-unique SVP is NP-hard (under randomized reductions) for $\gamma \leq 1+1/n^\eps$ for any constant $\eps > 0$~\cite{AD13}.

With some more work, we are also able to use our SVP reduction to derive the following search-to-decision reduction for CVP.

\begin{theorem}[CVP reduction]
\label{thm:CVPintro}
For any $\gamma = \gamma(n) \geq 1$ and $\ell = \ell(n) \geq 1$, there is a dimension-preserving (randomized) reduction from $\gamma^{n/\ell}$-CVP to $\gamma$-GapCVP that runs in time $n^{O(\ell)} \cdot \gamma^{O(n)} $.
\end{theorem}
This result is primarily interesting when $\ell$ is any constant and $\gamma \leq 1 + O(\log n/n)$, in which case the reduction runs in polynomial time and the search approximation factor is $\gamma^{O(n)} \leq \poly(n)$. But, it is still non-trivial for $1 < \gamma \leq 1 + \eps$ and $\ell \leq \eps n/\log n$, where $\eps > 0$ is some universal constant.

We actually show a (deterministic) $n^{O(\ell)}$-time reduction from $\gamma^{n/\ell}$-CVP to $\gamma$-GapCVP that works in the presence of a $\poly(n)$-SVP oracle. (See Theorem~\ref{thm:CVP}.) The above result then follows from instantiating this oracle via our SVP reduction.

Finally, we show deterministic reductions that achieve much worse parameters.

\begin{theorem}[Deterministic SVP reduction]
\label{thm:SVPdet}
For any $\gamma = \gamma(n) \geq 1$ and $p = p(n) \geq 2$, there is a deterministic dimension-preserving reduction from $\gamma'$-SVP to $\gamma$-GapSVP that runs in time $\poly(n) \cdot p$, where 
$
\gamma' := \gamma^{O(n^2/\log p)}
$.
\end{theorem}

\begin{theorem}[Deterministic CVP reduction]
\label{thm:CVPdet}
For any $\gamma = \gamma(n) \geq 1$ and $p = p(n) \geq 2$, there is a deterministic dimension-preserving reduction from $\gamma'$-CVP to $\gamma$-GapCVP that runs in time $\poly(n) \cdot p$, where 
$
\gamma' := \gamma^{O(n^2/\log p)}
$.
\end{theorem}
It is easy to see that our randomized reductions always give a better trade-off between the approximation factor and running time for non-trivial parameters. So, these new reductions are primarily interesting because they are deterministic and because they demonstrate additional potential approaches for future work in this area.

We note that all of our reductions are Cook reductions. (They make many oracle calls, sometimes adaptively.)

\subsection{Techniques}

\paragraph{SVP. } Our main SVP reduction works by finding a sublattice of the input lattice that has one relatively short vector but a significantly longer \scarequotes{second-shortest vector.} (I.e., we wish to find a sublattice of the input lattice that satisfies the promise of $\gamma$-unique SVP. See Definition~\ref{def:uSVP}.) To accomplish this, we use lattice sparsification, which was introduced by Khot~\cite{Khot05svp} and refined in a series of works~\cite{DadushK13,cvpp,DGStoSVP}. 

The \scarequotes{sparsified} sublattice is given by
\[
\lat' := \{\vec{y} \in \lat \, : \, \inner{\vec{z}, \basis^{-1}\vec{y}} \equiv 0 \bmod p \}
\; ,
\] where $p$ is some prime and $\vec{z} \in \Z_p^n$ is chosen uniformly at random. We would like to say that each short vector in $\lat$ will land in $\lat'$ with probability roughly $1/p$, independently of all other short vectors. Of course, if $\vec{x}, \vec{y} \in \lat$ and $\vec{x} = k \vec{y}$ for some $k \not\equiv 0 \bmod p$, then clearly $\vec{x} $ will land in $\lat'$ if and only if $\vec{y}$ does. So, we cannot have anything close to independence in this case. Luckily, \cite{DGStoSVP} shows that this is essentially \scarequotes{the only bad case.}

Specifically, a lattice vector $\vec{x} \in \lat$ is \emph{non-primitive} if $\vec{x} = k\vec{y}$ for some $k \geq 2$ and $\vec{y} \in \lat$. Otherwise, $\vec{x}$ is \emph{primitive}. (Note in particular that a shortest lattice vector is always primitive.) \cite{DGStoSVP} showed that sparsification behaves very nicely if we restrict our attention to primitive short lattice vectors. I.e., the distribution of short primitive vectors in the sparsified sublattice $\lat'$ is statistically close to the distribution that selects each short primitive vector from $\lat$ independently with probably $1/p$. (See Theorem~\ref{thm:sparsification} for the precise statement, which is taken directly from~\cite[Theorem 4.1]{DGStoSVP}.) 

So, let $\xi(\lat, r)$ be the number of primitive lattice points of length at most $r$. Suppose there exists some radius $r$ such that $\xi(\lat, \gamma r)$ is not much larger than $\xi(\lat, r)$.\footnote{It might be helpful to think of the heuristic that $\xi(\lat, \gamma r)/\xi(\lat, r) \approx \gamma^n$. This holds in the limit as $r \rightarrow \infty$, and it holds for random lattices in expectation.} Then, if we take $p \approx \xi(\lat, \gamma r)$, we can expect $\lat'$ to contain a primitive vector of length at most $r$ \emph{but no other primitive vectors of length less than $\gamma r$} with probability $\Theta(\xi(\lat, r)/\xi(\lat, \gamma r))$. In other words, with this probability, $\lat'$ will be a valid instance of $\gamma$-unique SVP with $\lambda_1(\lat') \leq r$, so that we can use an oracle for $\gamma$-unique SVP to find a non-zero lattice vector of length at most $r$.

The parameter $a$ in the reduction determines how large of a ratio $ \xi(\lat, \gamma r)/\xi(\lat, r)$ we are \scarequotes{willing to tolerate.} In particular, a simple proof shows that, for $a \geq n \log \gamma$, there is always a radius $r \leq \gamma^{n/a} \cdot \lambda_1(\lat)$ such that this ratio is bounded by $2^{O(a)}$. We can therefore obtain a valid $\gamma$-unique SVP instance with $\lambda_1(\lat') \leq r \leq \gamma^{n/a} \cdot \lambda_1(\lat)$ with probability at least $2^{-O(a)}$. So, our main reduction essentially works by sampling $\lat'$ repeatedly, a total of $2^{O(a)}$ times, and calling its $\gamma$-unique SVP oracle on each sampled sublattice $\lat'$.

\paragraph{CVP. } Our search-to-decision reduction for CVP is a simple \scarequotes{guided} variant of Babai's celebrated nearest-hyperplane algorithm~\cite{Bab86}. Babai's algorithm works by dividing the lattice into $(n-1)$-dimensional lattice hyperplanes and then restricting our search to the closest hyperplane to the target. However, it is possible that the closest lattice hyperplane does not actually contain very close lattice points. As a result, Babai's algorithm can lead to quite large approximation factors. 

So, we instead use our GapCVP oracle to \scarequotes{test many nearby hyperplanes} to find one that is guaranteed to contain a $\gamma$-approximate closest lattice point. By repeating this $n$ times over hyperplanes of progressively lower dimensions, we will find a $\gamma^n$-approximate closest vector to the target. To find a $\gamma^{n/\ell}$-approximate closest vector, we do the same thing with all nearby $(n-\ell)$-dimensional hyperplanes.

In order to make this algorithm efficient, we need to limit the number of hyperplanes that we must consider. This amounts to finding a short non-zero vector in the dual lattice. We can find such a vector by using our search-to-decision reduction for SVP (together with the known reduction from GapSVP to GapCVP~\cite{GMSS99}). Unfortunately, this costs us a factor of $\gamma^{O(n)}$ in the running time.

\paragraph{Deterministic reductions. } Our alternative deterministic search-to-decision reductions for SVP and CVP are very similar to the reduction from unique SVP to GapSVP in~\cite{LM09}. They essentially work by finding the coordinates of a short (or close) lattice vector \scarequotes{bit by bit.} I.e., in the CVP case, we first use our GapCVP oracle to compare the distance from the target to all lattice vectors whose last coordinate is even with its distance from all lattice vectors whose last coordinate is odd. If, say, the odd estimate is lower, then we restrict our attention to the lattice coset of all lattice vectors whose last coordinate is odd. We choose the remaining bits similarly, eventually obtaining the coordinates of a relatively close lattice vector. Our more general reductions follow from \scarequotes{working in base $p$ instead of base $2$.}

\subsection{Related work}

Some efficient dimension-preserving search-to-decision reductions were known for other lattice problems prior to this work. For example, Regev showed such a reduction for Learning with Errors, an important average-case lattice problem with widespread applications in cryptography~\cite{Reg09}. (Both the search and decision versions of LWE are average-case problems.) And, Liu, Lyubashevsky, and Micciancio implicitly use a search-to-decision reduction for Bounded Distance Decoding in their work~\cite{LiuLM06}. Finally, Aggarwal and Dubey showed how to use some of the ideas from~\cite{LM09} to obtain a search-to-decision reduction for unique SVP \cite{AD13}. While all of these works are quite interesting, they are concerned with promise problems, and not the two most important and natural  lattice problems, SVP and CVP.

More generally, this work can be seen as part of the ongoing study of the relationships between lattice problems under dimension-preserving reductions. By now, this area has become quite fruitful (e.g.,~\cite{Kan87,GMSS99, Micciancio08,LM09,DGStoSVP}). See~\cite{latticereductions} for a brief survey of well-known dimension-preserving reductions between various lattice problems.

Most prior work used sparsification to remove a relatively small number of \scarequotes{annoying} short vectors from a lattice without losing too many \scarequotes{good} short vectors  (e.g.,~\cite{Khot05svp,DadushK13,cvpp}). In our main SVP reduction, our goal is instead to remove \scarequotes{all but one} short vector. (Independent work of Bai, Wen, and Stehl{\'e} used sparsification in a similar way to reduce Bounded Distance Decoding to unique SVP~\cite{BWS16}.) To obtain our result, we rely heavily on the sparsification analysis of~\cite{DGStoSVP}, which is tighter and more general than prior work.

Interestingly, Kumar and Sivakumar used a procedure that is very similar to sparsification in their study of unique SVP, published in 2001~\cite{KS01}. Indeed, they repeatedly sparsify a lattice with $p = 2$ to obtain a sequence of sublattices such that at least one of these sublattices (1) contains a shortest non-zero vector of the original lattice; and (2) contains no other vectors of this length (up to sign, of course). However, the length of the \scarequotes{second-shortest vector} can be arbitrarily close to that of the shortest vector in their construction, even in a fixed dimension $n$. I.e., they use a restricted form of sparsification to effectively reduce $1$-SVP to $1$-unique SVP. Our main SVP reduction can be thought of as an updated version of their result. We use tools that were not available fifteen years ago to obtain a lower bound on the ratio between the shortest vector and the \scarequotes{second-shortest vector} that depends only on the dimension $n$. 

To prove hardness of $(1+1/\poly(n))$-unique SVP, Aggarwal and Dubey used the result of Kumar and Sivakumar to show a reduction from SVP to $\gamma$-unique SVP that works for a restricted subset of lattices~\cite{AD13}. In particular, Aggarwal and Dubey chose a set of lattices such that, over these lattices (1) SVP is NP-hard (as proven by Khot~\cite{Khot05svp}); and (2) this reduction yields $\gamma = 1+1/\poly(n)$. In contrast, we directly reduce $2$-SVP to $(1+1/\poly(n))$-unique SVP over \emph{all} lattices by using a much stronger (and unfortunately more complicated) form of Kumar and Sivakumar's reduction.

While the author knows of no other use of our specific variant of Babai's algorithm, we feel that it is quite natural and not particularly novel. For example, a similar idea was used in a different context by Micciancio~\cite[Corollary 7]{Micciancio08}. Our primary contribution on this front is the observation that this method gives a non-trivial search-to-decision reduction when the decision approximation factor is very small, and when it is combined with our SVP reduction.

We rely heavily on Lyubashevsky and Micciancio's dimension-preserving reduction from $\gamma$-unique SVP to $\gamma$-GapSVP~\cite{LM09}. Their result is necessary to prove Theorem~\ref{thm:SVPintro}, and our deterministic \scarequotes{bit-by-bit} SVP reduction is very similar to Lyubashevsky and Micciancio's reduction. The main difference between our deterministic SVP reduction and that of Lyubashevsky and Micciancio is that~\cite{LM09} work only with lattices that satisfy the promise of $\gamma$-unique SVP. They show that this promise is enough to guarantee that the $\gamma$-GapSVP oracle essentially behaves as an \emph{exact} GapSVP oracle. In contrast, our reduction works over general lattices, so we have to worry about accumulating error. (We also use a different method to \scarequotes{reduce the dimension of the lattice.})

\subsection{Directions for future work}

We view this paper as a first step towards a better understanding of the relationship between the search and decision variants of approximate SVP and CVP. In particular, we show that efficient, dimension-preserving search-to-decision reductions do in fact exist for approximation factors $\gamma > 1$. Prior to this work, one might have reasonably conjectured that such reductions do not exist for non-trivial parameters. But, our reductions lose quite a bit in the approximation factor, and the running times of our main reductions blow up quickly as the approximation factor increases. They are therefore primarily interesting for very small approximation factors $\gamma = 1+o(1)$. 

Results for such low values of $\gamma$ have sometimes led to similar results for larger approximation factors. For example, hardness of $\gamma$-GapSVP was originally proven for $\gamma = 1 + 2^{-\poly(n)}$~\cite{Ajtai-SVP-hard}, and then for $\gamma = 1+1/\poly(n)$~\cite{CN99}, before better inapproximability results were found~\cite{Mic01svp,Khot05svp,HRsvp}.  We therefore ask whether better search-to-decision reductions exist, and in particular, whether non-trivial efficient dimension-preserving reductions exist for larger approximation factors.

More specifically, we note that our main reductions are only efficient when the decision approximation factor is $\gamma = 1+ O(\log n/n)$ because their running time is proportional to $\gamma^{O(n)}$. This seems inherent to our technique in the case of SVP, and the CVP reduction suffers the same fate because it uses the SVP reduction as a subroutine. However, we see no reason why the running time should necessarily increase with the approximation factor, and this might simply be an artifact of our techniques. So, perhaps we can find reductions that do not have this problem. (One might try, for example, to eliminate the need for the SVP oracle in our CVP reduction.) Indeed, the reductions in Section~\ref{sec:worse} manage to avoid this pitfall, but they blow up the approximation factor much more and never actually outperform our main reductions.

In the other direction, we ask whether the search and decision versions of SVP and CVP can be separated in any way. I.e., can we show that, for some $\gamma > 1$, there is no efficient dimension-preserving reduction from $\gamma$-CVP to $\gamma$-GapCVP or no such reduction from $\gamma$-SVP to $\gamma$-GapSVP (under reasonable complexity-theoretic assumptions or even restrictions on the behavior of the reduction)? Can we find algorithms that solve the decision problems faster than our current search-based techniques allow (something more general than the rather specific examples mentioned in footnote~\ref{foot:searchdecisiongap})?
Of course, any such result would be a major breakthrough.

\subsection*{Acknowledgments}

I would like to thank Divesh Aggarwal, Huck Bennett, Zvika Brakerski, Daniel Dadush, Daniele Micciancio, and Oded Regev, all of whom provided helpful commentary on earlier drafts of this paper. In particular, Divesh noted that our SVP reduction shows hardness of unique SVP (Corollary~\ref{cor:uSVPhard}), and Daniele noted that it implied a reduction from $\sqrt{n}$-CVP to $\gamma$-GapSVP (Corollary~\ref{cor:cvptosvp}). I also thank the anonymous reviewers for their helpful comments.

\section{Preliminaries}

We write $\log x$ for the logarithm of $x$ in base $2$. We write $\|\vec{x}\| $ for the Euclidean norm of $\vec{x} \in \R^n$. We omit any mention of the bit length of the input throughout. In particular, all of our algorithms take as input vectors in $\R^n$ (with some reasonable representation) and run in time $f(n) \cdot \poly(m)$ for some $f$, where $m$ is the maximal bit length of an input vector. We are primarily interested in the dependence on $n$, so we suppress the factor of $\poly(m)$.

\subsection{Lattice basics}

A rank $d$ lattice $\lat\subset \R^n$ is the set of all integer linear combinations of $d$ linearly independent vectors $\basis = (\vec{b}_1, \ldots, \vec{b}_d)$. $\basis$ is called a basis of the lattice and is not unique. 
We write $ \lat(\basis)$ to signify the lattice generated by $\basis$.  
By taking the ambient space to be $\spn(\lat)$, we can implicitly assume that a lattice has full rank $n$, and we therefore will often implicitly assume that $d = n$.

The dual lattice is
\[ \lat^* := \{ \vec{w} \in \spn(\lat) : \forall \vec{y} \in \lat, \inner{\vec{w},\vec{y}} \in \mathbb{Z} \}\;. \]
Similarly, the dual basis $\basis^* := \basis (\basis^T\basis)^{-1} = (\vec{b}_1^*,\ldots, \vec{b}_d^*)$ is the unique list of vectors in $\spn(\lat)$ satisfying $\inner{ \vec{b}_i^*, \vec{b}_j} = \delta_{i,j} $. $\lat^*$ is itself a rank $d$ lattice with basis $\basis^*$.

We write $\lambda_1(\lat) := \min_{\vec{x} \in \lat \setminus \{\vec0 \}} \| \vec{x}\|$ for the length of the shortest non-zero vector in the lattice. Similarly, we write $\lambda_2(\lat) := \min \{ r > 0 \ : \ \dim(\spn(\lat \cap r B_2^n)) \geq 2 \}$ for the length of the shortest lattice vector that is linearly independent from a lattice vector of length $\lambda_1(\lat)$. For any point $\vec{t} \in \R^n$, we write $\dist(\vec{t}, \lat) := \min_{\vec{x} \in \lat} \| \vec{x} - \vec{t}\|$ for the distance between $\vec{t}$ and $\lat$. The covering radius $\mu(\lat) := \max_{\vec{t} \in \spn(\lat)} \dist(\vec{t}, \lat)$ is the maximal such distance achievable in the span of the lattice.

The following two bounds will be useful.\footnote{We note that tighter bounds exist for the number of lattice points in a ball of radius $r$~\cite{KL78,PS09}, but we use the bound of~\cite{BHW93} because it is simpler. Using a tighter bound here would improve the hidden constants in the exponents of our running times.}

\begin{theorem}[{\cite[Theorem 2.1]{BHW93}}]
\label{thm:lat-pt-bnd}
For any lattice $\lat \subset \R^n$ and $r > 0$,
\[
|\set{\vec{y} \in \lat: \|\vec{y}\| \leq r \lambda_1(\lat)}| \leq
2\ceil{2r}^n-1 \text{.} \]
\end{theorem}

\begin{lemma}[{\cite[Theorem 2.2]{banaszczyk}}]
\label{lem:coveringlambda1}
For any lattice $\lat \subset \R^n$, $ \lambda_1(\lat^*) \cdot \mu(\lat) \leq n/2$.
\end{lemma}

We derive a simple though rather specific corollary of Lemma~\ref{lem:coveringlambda1} that we will use twice. The corollary says that a dual vector $\vec{w} \in \lat^* \setminus \{\vec0\}$ that is relatively short, $\| \vec{w} \| \leq \gamma \cdot \lambda_1(\lat^*)$, can be used to partition $\lat$ into $(n-1)$-dimensional lattice hyperplanes, such that the closest vector to any target $\vec{t}$ must lie in one of the $O(\gamma n)$ hyperplanes closest to $\vec{t}$.

\begin{corollary}
\label{cor:hyperplanes}
For any lattice $\lat \subset \R^n$ with basis $(\vec{b}_1,\ldots, \vec{b}_n)$ and associated dual basis $(\vec{b}_1^*,\ldots, \vec{b}_n^*)$, $\gamma \geq 1$, and any target $\vec{t} \in \R^n$, if $\|\vec{b}_1^*\| \leq \gamma \cdot \lambda_1(\lat^*)$, then any closest lattice vector to $\vec{t}$ must lie in a lattice hyperplane $\lat' + i \vec{b}_1$, where $\lat' := \lat(\vec{b}_2,\ldots, \vec{b}_n)$ and $i$ is an integer with $|i - \inner{\vec{b}_1^*, \vec{t}}| \leq \gamma n/2$.
\end{corollary}
\begin{proof}
Let $\vec{y} \in \lat$ be a closest lattice vector to $\vec{t}$. It follows from the definition of a lattice that $\vec{y} \in \lat' + i\vec{b}_1$ for some integer $i= \inner{\vec{b}_1^*, \vec{y}}$. We have
\[
|i - \inner{\vec{b}_1^*, \vec{t}}| = |\inner{\vec{b}_1^*, \vec{y} - \vec{t}}| \leq \|\vec{b}_1^*\| \|\vec{y} - \vec{t}\| \leq \gamma \lambda_1(\lat^*) \cdot \mu(\lat) \leq \gamma n/2 \; ,
\]
where we have used Lemma~\ref{lem:coveringlambda1}.
\end{proof}

\subsection{LLL-reduced bases}

Given a basis, $\basis = (\vec{b}_1,\ldots, \vec{b}_n)$, we define its Gram-Schmidt orthogonalization $(\gs{\vec{b}}_1,\ldots, \gs{\vec{b}}_n)$ by
\[  \gs{\vec{b}}_i = \pi_{\{\vec{b}_1, \ldots, \vec{b}_{i-1} \}^\perp}(\vec{b}_i) \; , \]
and the Gram-Schmidt coefficients $\mu_{i,j}$ by
\[ \mu_{i,j}= \frac{\inner{\vec{b}_i, \gs{\vec{b}}_j}}{\length{\gs{\vec{b}}_j}^2}\; . \]
Here, $\pi_A$ represents orthogonal projection onto the subspace $A$ and $\{\vec{b}_1, \ldots, \vec{b}_{i-1} \}^\perp$ denotes the subspace of vectors in $\R^n$ that are orthogonal to $\vec{b}_1, \ldots, \vec{b}_{i-1}$.

\begin{definition}
A basis $\basis = (\vec{b}_1,\ldots, \vec{b}_n)$ is \emph{LLL-reduced} if 
\begin{enumerate}
\item for $1 \leq j < i \leq n$, $|\mu_{i,j}| \leq 1/2$; and
\item for $2 \leq i \leq n$, $\|\gs{\vec{b}}_i\|^2 \geq (3/4 - \mu_{i,i-1}^2) \cdot \| \gs{\vec{b}}_{i-1}\|^2$.
\end{enumerate}
\end{definition}

\begin{theorem}[{\cite{LLL82}}]
\label{thm:LLL}
There exists an efficient algorithm that takes as input a (basis for) a lattice and outputs an LLL-reduced basis for the lattice.
\end{theorem}

\begin{lemma}
\label{lem:LLLcoordinates}
For any lattice $\lat \subset \R^n$ with LLL-reduced basis $\basis =(\vec{b}_1,\ldots, \vec{b}_n)$ and $\vec{y} = \sum a_i \vec{b}_i \in \lat$, we have
\[
|a_i| \leq 2^{3n/2-i} \cdot \frac{\|\vec{y}\|}{\lambda_1(\lat)}
\; ,
\]
for all $i$.
\end{lemma}
\begin{proof}
It follows immediately from the definition of an LLL-reduced basis that $\|\gs{\vec{b}}_i\| \geq \|\vec{b}_1\| /2^{i/2} \geq \lambda_1(\lat)/2^{n/2}$ for all $i$. For each $i$, we have
\[
\|\vec{y}\| \geq \sum_{j=1}^{n} |a_j\mu_{j,i}|\cdot \|\gs{\vec{b}}_{i}\| = \Big( |a_i| -\sum_{j=i+1}^n|a_j\mu_{j,i}|\Big) \cdot \|\gs{\vec{b}}_{i}\| \geq \Big( |a_i| - \frac{1}{2}\sum_{j=i+1}^n|a_j|\Big) \cdot 2^{-n/2} \cdot \lambda_1(\lat)
\; .
\]
In particular, $ |a_n| \leq  2^{n/2} \cdot \|\vec{y}\|/ \lambda_1(\lat)$. We assume for induction that $|a_j| \leq 2^{3n/2-j} \cdot \|\vec{y}\|/\lambda_1(\lat)$ for all $j$ with $i < j \leq n$. Then, plugging in to the above, we have
\[
\|\vec{y}\| \geq  |a_i| \cdot 2^{-n/2} \cdot \lambda_1(\lat) - \sum_{j=i+1}^n2^{n-j-1} \| \vec{y}\| \geq |a_i| \cdot 2^{-n/2} \cdot \lambda_1(\lat) - 2^{n - i-1} \cdot \|\vec{y}\|
\; .
\]
The result follows by rearranging.
\end{proof}

\begin{lemma}[\cite{Bab86}]
\label{lem:babai_covering_radius}
If $\basis = (\vec{b}_1,\ldots, \vec{b}_n)$ is an LLL-reduced basis, then $\mu(\lat) \leq \sqrt{n} 2^{n/2-1} \cdot \|\gs{\vec{b}}_n\|$.
\end{lemma}

\subsection{Lattice problems}

We now list the computational problems that concern us. All of the below definitions are standard.

\begin{definition}
For any parameter $\gamma =\gamma(n) \geq 1$, $\gamma$-SVP (the Shortest Vector Problem) is the search problem defined as follows: The input is a basis $\basis$ for a lattice $\lat \subset \R^n$. The goal is to output a lattice vector $\vec{x}$ with $0 < \length{\vec{x}} \leq \gamma \lambda_1(\lat)$.
\end{definition}

\begin{definition}
For any parameter $\gamma = \gamma(n) \geq 1$, $\gamma$-CVP (the Closest Vector Problem) is the search problem defined as follows: The input is a basis $\basis$ for a lattice $\lat \subset \R^n$ and a target vector $\vec{t} \in \R^n$. The goal is to output a lattice vector $\vec{x}$ with $\length{\vec{x} - \vec{t}} \leq \gamma \dist(\vec{t}, \lat)$.
\end{definition}

\begin{definition}
For any parameter $\gamma = \gamma(n) \geq 1$, the decision problem $\gamma\text{-}\problem{GapSVP}$ is defined as follows: The input is a basis $\basis$ for a lattice $\lat \subset \R^n$ and a number $d > 0$. The goal is to output yes if $\lambda_1(\lat) < d$ and no if $\lambda_1(\lat) \geq \gamma \cdot d$.
\end{definition}

\begin{definition}
For any parameter $\gamma = \gamma(n) \geq 1$, the decision problem $\gamma\text{-}\problem{GapCVP}$ is defined as follows: The input is a basis $\basis$ for a lattice $\lat \subset \R^n$, a target $\vec{t} \in \R^n$, and a number $d > 0$. The goal is to output yes if $\dist(\vec{t},\lat) < d$ and no if $\dist(\vec{t},\lat) \geq \gamma \cdot d$.
\end{definition}

\begin{definition}
\label{def:uSVP}
For any parameter $\gamma = \gamma(n) \geq 1$, $\gamma$-uSVP (the Unique Shortest Vector Problem) is the search promise problem defined as follows: The input is a basis $\basis$ for a lattice $\lat \subset \R^n$ with $\lambda_2(\lat) \geq \gamma(n) \cdot \lambda_1(\lat)$. The goal is to output a lattice vector $\vec{x}$ with $\length{\vec{x}} = \lambda_1(\lat)$.
\end{definition}

\subsection{Known results}

We will need the following known reductions and hardness results.

\begin{theorem}[\cite{Khot05svp}]
\label{thm:SVPhard}
For any constant $\gamma \geq 1$, $\gamma$-GapSVP (and therefore $\gamma$-SVP) is NP-hard under randomized reductions.
\end{theorem}

\begin{theorem}[\cite{GMSS99}]
\label{thm:SVPtoCVP}
For any $\gamma \geq 1$, there is an efficient dimension-preserving reduction from $\gamma$-GapSVP to $\gamma$-GapCVP (and from $\gamma$-SVP to $\gamma$-CVP).
\end{theorem}

\begin{theorem}[{\cite[Theorem 6.1]{LM09}}]
\label{thm:uSVPtoGapSVP}
For any $1 \leq \gamma(n) \leq \poly(n)$, there is an efficient dimension-preserving reduction from $\gamma$-uSVP to $\gamma$-GapSVP.
\end{theorem}

\begin{theorem}[{\cite[Theorem 4.2]{MicciancioBook}}]
\label{thm:cvptosvp}
There  is an efficient reduction from $\sqrt{n}$-CVP to $\sqrt{2}$-SVP.
Furthermore, all of the oracle calls of the reduction are made in dimension $n+1$, where $n$ is the input dimension.
\end{theorem}

Reductions like that of Theorem~\ref{thm:cvptosvp} that increase the dimension by one are good enough for nearly all applications of perfectly dimension-preserving reductions. But, we can use a simple idea to convert Theorem~\ref{thm:cvptosvp} into a reduction that preserves the dimension exactly. (Micciancio uses essentially the same trick in the proof of~\cite[Corollary 7]{Micciancio08}.)

\begin{corollary}
\label{cor:cvptosvp}
There is a dimension-preserving efficient reduction from $\sqrt{n}$-CVP to $\sqrt{2}$-SVP.
\end{corollary}
\begin{proof}
	On input $\vec{t} \in \R^n$ and a lattice $\lat \subset \R^n$, the reduction first uses its SVP oracle to find a vector $\vec{b}_1^* \in \lat^*$ in the dual with $0 < \length{\vec{b}_1^*} < 2\lambda_1(\lat^*)$. Let $\basis^* := (\vec{b}_1^*, \ldots, \vec{b}_n^*)$ be a basis for $\lat^*$, and let $\basis = (\vec{b}_1,\ldots, \vec{b}_n)$ be the associated primal basis. (Since $\vec{b}_1^*$ is a primitive lattice vector, it is always possible to find such a basis.) Let $\lat' := \lat(\vec{b}_2,\ldots, \vec{b}_n)$ be the lattice generated by $\basis$ with the first basis vector removed. Finally, let $a := \inner{\vec{b}_1^*, \vec{t}}$. For $i = \floor{a}-n, \ldots, \ceil{a} + n$, the reduction runs the procedure from Theorem~\ref{thm:cvptosvp} on input $\vec{t} - i \vec{b}_1$ and $\lat'$, receiving as output $\vec{y}_i \in \lat'$. The reduction then simply outputs a closest vector to $\vec{t}$ amongst the vectors $\vec{y}_i + i \vec{b}_1 \in \lat$.
	
It is clear that the reduction is efficient. Furthermore, note that the reduction only uses the procedure from Theorem~\ref{thm:cvptosvp} with $(n-1)$-dimensional input. (Formally, we must project $\lat'$ and $\vec{t} - i \vec{b}_1$ onto $\spn(\lat') \cong \R^{n-1}$.) Since that procedure increases dimension by one, this new reduction preserves dimension. For correctness, we note that Corollary~\ref{cor:hyperplanes} implies that there is a closest vector to $\vec{t}$ in one of the lattice hyperplanes $\lat' + i \vec{y}_i$. The result then follows from Theorem~\ref{thm:cvptosvp}.
\end{proof}

\subsection{A note on decision and estimation}

Formally, we consider \emph{gapped decision problems}, which take as input a number $d > 0$ and some additional input $I$ and require us to output YES if $f(I) \leq d$ and NO if $f(I) > \gamma d$, where $f$ is some function and $\gamma$ is the approximation factor. (For example, $I$ may be some representation of a lattice and $f(I)$ may be the length of the shortest vector in the lattice.) However, it is sometimes convenient to work with \emph{estimation problems}, which take only $I$ as input and ask for a numerical output $\tilde{d}$ with $f(I) \leq \tilde{d} \leq \gamma f(I)$. 

For the specific problems that we consider (and most \scarequotes{sufficiently nice} problems), the estimation variants are equivalent to the gapped decision problems as long as the lattice is \scarequotes{represented reasonably} by the input. For example, if $f(I)$ can be represented as a string of length at most $\poly(|I|)$ (e.g., $f(I)$ might be a rational number with bounded numerator and denominator), then we can use binary search and a gapped decision oracle to estimate $f(I)$ efficiently. This is true, for example, whenever the input is interpreted as a list of vectors with rational coordinates, using the standard representation of rational numbers. (See~\cite{MicciancioBook} for a careful discussion of this and related issues in the context of lattice problems.) We therefore make no distinction between gapped decision problems and estimation problems in the sequel, without worrying about the specific form of our input, or more generally, the specific representation of numbers.

\section{Reducing SVP to uSVP (and GapSVP) via sparsification}

\label{sec:SVP}

\subsection{Sparsification}

For a lattice $\lat \subset \R^n$, we write $\lat_{\mathrm{prim}}$ for the set of all primitive vectors in $\lat$, and $\xi(\lat, r) := |\lat_{\mathrm{prim}} \cap r B_2^n|/2 $ for the number of primitive lattice vectors contained in a (closed) ball of radius $r$ around the origin (counting $\vec{x}$ and $-\vec{x}$ as a single vector). The following theorem from~\cite{DGStoSVP} shows that sparsification behaves nicely with respect to primitive vectors, which is enough for our use case.

\begin{theorem}[{\cite[Theorem 4.1]{DGStoSVP}}]
\label{thm:sparsification}
For any lattice $\lat \subset \R^n$ with basis $\basis$, primitive lattice vectors $\vec{y}_0, \vec{y}_1,\ldots, \vec{y}_N \in \lat_{\mathrm{prim}}$ with $\vec{y}_i \neq \pm \vec{y}_0$ for all $i > 0$, and prime $p \geq 101$, if $\xi(\lat, \length{\vec{y}_i}) \leq p/(20 \log p)$ for all $i$, then
\[
\frac{1}{p} - \frac{N}{p^2} \leq \Pr\big[\inner{\vec{z}, \basis^{-1}\vec{y}_0} \equiv 0 \bmod p \text{ and } \inner{\vec{z}, \basis^{-1}\vec{y}_i} \not\equiv 0 \bmod p \ \forall i > 0\big] \leq \frac{1}{p} \; ,
\]
where $\vec{z} \in \Z_p^n$ is chosen uniformly at random.
\end{theorem}

From this, we can immediately derive the following proposition, which is a slight variant of~\cite[Proposition 4.2]{DGStoSVP}.

\begin{proposition}
\label{prop:sparsifier}
There is an efficient (randomized) algorithm that takes as input a basis $\basis$ for a lattice $\lat \subset \R^n$ and a prime $p \geq 101$ 
and outputs a full-rank sublattice $\lat' \subseteq \lat$ such that for any $r_1, r_2$, with $\lambda_1(\lat) \leq r_1 \leq r_2 < p\lambda_1(\lat)$ and $\xi(\lat, r_2) \leq p/(20 \log p)$, we have
\[
\Pr[\lambda_1(\lat') \leq r_1 \mathrm{\ and\ } \lambda_2(\lat') > r_2] \geq \frac{\xi(\lat, r_1)}{p} \cdot \Big(1 - \frac{\xi(\lat, r_2)}{p}
\Big)
\; . 
\]
\end{proposition}
\begin{proof}
The algorithm takes as input a basis $\basis$ for a lattice $\lat \subset \R^n$. It then samples $\vec{z} \in \Z_p^n$ uniformly at random and outputs 
\[
\lat' := \{ \vec{y} \in \lat \ : \ \inner{\vec{z}, \basis^{-1}\vec{y}} \equiv 0 \bmod p \}
\; .
\]

Let $N := \xi(\lat, r_2) \leq p/(20 \log p)$, and let $\vec{y}_1,\ldots, \vec{y}_{N} \in \lat$ be the $N$ unique primitive vectors in $\lat$ satisfying $\| \vec{y}_1 \| \leq \cdots \leq \|\vec{y}_N\| \leq r_2$ (taking only one vector from each pair $\pm \vec{y}$). 
Note that we have $\lambda_1(\lat') \leq r_1$ and $\lambda_2(\lat') > r_2$ if and only if $\vec{y}_i \in \lat'$ for some $i \leq \xi(\lat, r_1)$ and $\vec{y}_j \notin \lat'$ for all $j \neq i$. (Here, we have used the fact that $r_2 < p\lambda_1(\lat)$ to guarantee that vectors of the form $p\vec{y}_i \in \lat'$ do not cause $\lambda_2(\lat')$ to be less than $r_2$.) 

Applying Theorem~\ref{thm:sparsification}, we see that this happens with probability at least $1/p - (N-1)/p^2 > 1/p - N/p^2$ for any fixed $i$. The result follows by noting that these are disjoint events, so that the probability that at least one of these events occurs is the sum of their individual probabilities, which is at least
\[
\xi(\lat, r_1) \cdot \Big(\frac{1}{p} - \frac{N}{p^2} \Big) = \frac{\xi(\lat, r_1)}{p} \cdot \Big(1 - \frac{\xi(\lat, r_2)}{p} \Big)
\; ,
\]
as needed.
\end{proof}

\subsection{The reduction}

We now present the main step in our search-to-decision reduction for SVP.

\begin{theorem}
	\label{thm:SVP}
	For any $\gamma = \gamma(n) \geq 1$ and $a = a(n) \geq \log(n+1)$, there is a dimension-preserving (randomized) reduction from $\gamma^{n/a}$-SVP to $\gamma$-uSVP that runs in time $2^{O(a)} \cdot \gamma^{O(n)}$.
\end{theorem}
\begin{proof}
	We may assume without loss of generality that $a \leq n/2$, since the result is trivial for larger $a$. (There are known $2^{O(n)}$-time algorithms for SVP~\cite{AKS01,ADRS15}.) We may also assume that $2^a > \gamma^n$, since this does not affect the asymptotic running time.  Let $k := \ceil{4^{ a n/(n-a)}} = 2^{O(a)}$.
	
	On input a lattice $\lat \subset \R^n$, the reduction does the following $k$ times. For $i = 0, \ldots, \ell := \floor{n/a}$, let $p_i$ be a prime with $2 k^{i+1} < p_i < 4 k^{i+1}$. The reduction calls the procedure from Proposition~\ref{prop:sparsifier} with input $\lat$ and $p_i$, receiving as output $\lat_i$. It then calls its uSVP oracle on each $\lat_i$, receiving as output $\vec{x}_i$. Finally, it simply outputs a shortest non-zero $\vec{x}_i$.
	
	It is clear that the reduction runs in time $\poly(n) \cdot k = 2^{O(a)}$, as needed.\footnote{The reader might notice that the theorem quotes a running time of $2^{O(a)} \cdot \gamma^{O(n)}$ instead of just $2^{O(a)}$. Note that this looser bound on the running time is exactly what allowed us to assume $2^a > \gamma^n$ above. Equivalently, we could simply require $a > n \log \gamma$ in the theorem statement and achieve a running time of $2^{O(a)}$. But, we wish to avoid misunderstanding by explicitly stating that the running time is at least $\gamma^{O(n)}$ in the theorem statement.} For each $i$, let $r_i$ be minimal such that $\xi(\lat, r_i) \geq k^i$. In particular, $r_0 = \lambda_1(\lat)$. And, recalling the definition of $\xi$, we have
\begin{align*}
	|\lat \cap r_\ell B_2^n| &> 2 \xi(\lat, r_\ell) 
	\geq 2k^\ell 
	> 2\cdot (4^{ a n/(n-a)})^{n/a-1} 
	= 2 \cdot 4^n
	\; .
\end{align*}
	So, applying Theorem~\ref{thm:lat-pt-bnd}, we have that $r_\ell/r_0 = r_\ell/\lambda_1(\lat) > 2$. 
	
	Therefore, there exists an $i$ such that $r_{i+1}/r_i > 2^{1/\ell} \geq 2^{a/n} > \gamma$. Let $j$ be minimal such that $r_{j+1}/r_j > \gamma$. In particular, this means that $\xi(\lat, \gamma r_j) < k^{j+1}$ and $\gamma r_{j} \leq 2\gamma  \lambda_1(\lat) < p_j\lambda_1(\lat)$. So, we may apply Proposition~\ref{prop:sparsifier} to obtain
	\begin{align*}
	\Pr[\lambda_1(\lat_j) \leq r_j \mathrm{\ and\ } \lambda_2(\lat_j) > \gamma r_j] &\geq \frac{\xi(\lat, r_j)}{p_j} - \frac{\xi(\lat, r_j) \xi(\lat, \gamma r_j)}{p_j^2}\\
	&> \frac{k^j}{p_j} \cdot \Big(1 - \frac{k^{j+1}}{p_j} \Big)\\
	&> \frac{1}{2k}
	\; .
	\end{align*}
	Therefore, after running the above procedure $k$ times, the algorithm will output a non-zero vector of length at most $r_i$ with at least some positive constant probability.
	
	Finally, by the definition of $r_j$, we have $ r_j/\lambda_1(\lat) = r_j/r_0 \leq \gamma^j \leq \gamma^\ell \leq \gamma^{n/a}$. Therefore, the algorithm outputs a $\gamma^{n/a}$-approximate shortest vector with at least constant probability, as needed.
\end{proof}

\subsection{Corollaries}
\label{sec:cors}

From this, we derive some immediate corollaries. The first is our main SVP result.

\begin{proof}[Proof of Theorem~\ref{thm:SVPintro}]
We may assume without loss of generality that $\gamma \leq 2$, since otherwise the result is trivial as there are known $2^{O(n)}$-time algorithms for SVP. Therefore, by Theorem~\ref{thm:uSVPtoGapSVP}, there is an efficient dimension-preserving reduction from $\gamma$-uSVP to $\gamma$-GapSVP. The result then follows from Theorem~\ref{thm:SVP}.
\end{proof}

By combining Theorem~\ref{thm:SVP} with Corollary~\ref{cor:cvptosvp}, we obtain the following reduction from $\sqrt{n}$-CVP to $\gamma$-uSVP (and therefore $\gamma$-GapSVP).

\begin{corollary}
\label{cor:CVPtoGapSVP}
For any $\gamma = \gamma(n) \geq 1$, there is a dimension-preserving (randomized) reduction from $\sqrt{n}$-CVP to $\gamma$-uSVP that runs in time $\poly(n) \cdot \gamma^{O(n)}$. 

Similarly, there is a dimension-preserving (randomized) reduction from $\sqrt{n}$-CVP to $\gamma$-GapSVP with the same running time.
\end{corollary}
\begin{proof}
Setting $a := 2n \log(\gamma)+\log (n+1)$ in Theorem~\ref{thm:SVP} gives a reduction from $\sqrt{2}$-SVP to $\gamma$-uSVP with the claimed running time. The first result then follows from Corollary~\ref{cor:cvptosvp}. The second result follows from Theorem~\ref{thm:uSVPtoGapSVP}.
\end{proof}

We also obtain an alternative proof of the hardness of $(1+1/\poly(n))$-uSVP, as originally shown by Aggarwal and Dubey~\cite{AD13}.

\begin{corollary}
\label{cor:uSVPhard}
For any constant $\eps > 0$, $(1+1/n^\eps)$-uSVP is NP-hard (under randomized reductions).
\end{corollary}
\begin{proof}
For $\gamma \leq 1+ O(\log n/n)$, taking $a := n \log(\gamma) + \log(n+1)$ in Theorem~\ref{thm:SVP} gives a polynomial-time reduction from $2$-SVP to $\gamma$-uSVP. It then follows from Theorem~\ref{thm:SVPhard} that $\gamma$-uSVP is NP-hard (under randomized reductions).

The full result then follows by noting that there is a simple reduction from $(1+1/n^\eps)$-uSVP to $(1+1/n)$-uSVP for any constant $\eps \in (0,1)$. In particular, given input $\lat \subset \R^n$ with basis $\basis := (\vec{b}_1,\ldots, \vec{b}_n)$, let $N := \ceil{n^{1/\eps}} = \poly(n)$, and let $r := 3 \|\vec{b}_1\| > 2 \lambda_1(\lat)$. Let $\lat' := \lat(\vec{b}_1,\ldots, \vec{b}_n, r\vec{e}_{n+1}, \ldots, r\vec{e}_{N}) \subset \R^{N}$ be the rank $N$ lattice obtained by \scarequotes{adding $N -n$ perpendicular vectors of length $r$ to $\lat$.} The result follows by noting that $N^\eps \geq n$ so that $\lat'$ is a valid instance of $(1+1/N^\eps)$-uSVP if $\lat$ is a valid instance of $(1+1/n)$-uSVP, and the two instances have the same solution.
\end{proof}

Finally, we note that a reduction to GapSVP immediately implies a reduction to GapCVP, by Theorem~\ref{thm:SVPtoCVP}. We will need this in the next section.

\begin{corollary}
\label{cor:SVPtoGapCVP}
For any $\gamma = \gamma(n) \geq 1$ and $a = a(n) \geq \log(n+1)$, there is a dimension-preserving (randomized) reduction from $\gamma^{n/a}$-SVP to $\gamma$-GapCVP that runs in time $2^{O(a)} \cdot \gamma^{O(n)}$.
\end{corollary}
\begin{proof}
Combine Theorem~\ref{thm:SVPintro} with Theorem~\ref{thm:SVPtoCVP}.
\end{proof}

\section{Reducing CVP to GapCVP}
\label{sec:CVP}

\begin{theorem}
\label{thm:CVP}
For any $\gamma = \gamma(n) \geq 1$, $h = h(n) \geq 1$, and integer $\ell = \ell(n) \geq 1$, there is a (deterministic) algorithm with access to a $\gamma$-GapCVP oracle and a $h$-SVP oracle that solves $\gamma^{n/\ell}$-CVP in time $( \poly(n) \cdot  h)^\ell$. Furthermore, the dimension of the algorithm's oracle calls never exceeds the dimension of the input lattice.
\end{theorem}
\begin{proof}
We show how to handle the case $\ell =1$ and then describe how to extend the result to arbitrary $\ell$. On input a lattice $\lat \subset \R^n$ and $\vec{t} \in \R^n$, the algorithm behaves as follows. If $n = 1$, then it solves the CVP instance directly. Otherwise, it first uses 
its SVP oracle to find a dual vector $\vec{b}_1^* \in \lat^*$ with $\length{\vec{b}_1^*} \leq h \cdot \lambda_1(\lat^*)$. Let $\vec{b}_2^*, \ldots, \vec{b}_n^* \in \lat^*$ such that $(\vec{b}_1^*,\ldots, \vec{b}_n^*)$ is a basis of $\lat^*$, and let $(\vec{b}_1,\ldots, \vec{b}_n) \in \lat$ be the associated basis of the primal. (This is always possible if $\vec{b}_1^*$ is primitive in $\lat^*$. If $\vec{b}_1^*$ is not primitive, then we can simply replace it with a primitive vector that is a scalar multiple of $\vec{b}_1^*$.)

Next, let $a := \inner{\vec{b}_1^*, \vec{t}}$ and $\lat' := \lat(\vec{b}_2,\ldots, \vec{b}_n)$. Then, for $i = \floor{a - h \cdot n}, \ldots, \ceil{a + h \cdot n}$, the algorithm uses its GapCVP oracle to compute $d_i$ such that $\dist(\vec{t} - i \cdot \vec{b}_1, \lat') \leq d_i \leq \gamma \cdot \dist(\vec{t} - i \cdot \vec{b}_1, \lat')$. The algorithm then picks an index $i$ such that $d_i$ is minimal and calls itself recursively on input $\lat'$ and $\vec{t} - i \cdot \vec{b}_1$, receiving as output $\vec{y} \in \lat'$. Finally, it outputs $\vec{y} + i \vec{b}_1 \in \lat$.

It is clear that the running time is as claimed. By Corollary~\ref{cor:hyperplanes}, there must exist some $j$ such that $\dist(\vec{t} - j \cdot \vec{b}_1, \lat') = \dist(\vec{t}, \lat)$, so that $d_j \leq \gamma \dist(\vec{t}, \lat)$. Therefore, $d_i \leq \gamma \cdot \dist(\vec{t}, \lat)$, and $\dist(\vec{t} - i \cdot \vec{b}_1, \lat') \leq \gamma \cdot \dist(\vec{t}, \lat)$. The result then follows from induction on the dimension $n$.

To handle arbitrary $\ell \geq 1$, the algorithm simply tries all recursive paths up to depth $\ell$ and chooses the path that yields the lowest approximate distance according to its $\gamma$-GapCVP oracle. Note that there are at most $(2hn+2)^\ell = ( \poly(n) \cdot  h)^\ell$ such paths, so the running time is as claimed.
\end{proof}

We obtain our main CVP reduction by combining the above result with our SVP reduction.

\begin{proof}[Proof of Theorem~\ref{thm:CVPintro}]
We may assume without loss of generality that $\ell$ is an integer. 

We can instantiate the SVP oracle required in Theorem~\ref{thm:CVP} above by using Corollary~\ref{cor:SVPtoGapCVP}. In particular, taking $a := n \log \gamma + \log(n+1)$ in Corollary~\ref{cor:CVPtoGapSVP} gives a reduction from $2$-SVP to $\gamma$-GapCVP that runs in time $\poly(n) \cdot \gamma^{O(n)}$. By using this reduction to instantiate the $h$-SVP oracle in Theorem~\ref{thm:CVP} with $h = 2$, we get a dimension-preserving reduction from $\gamma^{n/\ell}$-CVP to $\gamma$-CVP that runs in time $n^{O(\ell)} \cdot \gamma^{O(n)}$, as needed.
\end{proof}

\section{Deterministic reductions}
\label{sec:worse}

We now show deterministic search-to-decision reductions for SVP and CVP that achieve significantly worse parameters. Both reductions use the same basic idea, which is essentially to \scarequotes{find the coordinates of a short (or close) lattice point bit-by-bit.}

\subsection{The deterministic CVP reduction}

We present the CVP reduction first because it is simpler.

\begin{proof}[Proof of Theorem~\ref{thm:CVPdet}]
We may assume without loss of generality that $p$ is an integer and $\gamma^{2} < p$, since the result is trivial for larger $\gamma$.

On input a lattice $\lat \subset \R^n$ and $\vec{t} \in \R^n$, the reduction behaves as follows. If $n = 1$, then it solves the CVP instance directly. Otherwise, it first uses the procedure from Theorem~\ref{thm:LLL}, to compute  an LLL-reduced basis $\basis = (\vec{b}_1,\ldots, \vec{b}_n)$ for $\lat$. It then finds the $n$th coordinate of a close lattice vector to $\vec{t}$ \scarequotes{in base $p$,} as follows. Let $\vec{t}_0 = \vec{t}$, and let $\lat_i := \lat(\vec{b}_1,\ldots, \vec{b}_{n-1}, p^{i} \cdot \vec{b}_n)$ for all $i$. 
For $i = 0,\ldots, \ell-1$, with $\ell \geq 1$ to be set in the analysis, the reduction uses its $\gamma$-GapCVP oracle to compute $d_{i,0}, \ldots, d_{i,p-1}$ such that $\dist(\vec{t}_i- j p^{i} \cdot \vec{b}_n, \lat_{i+1}) \leq  d_{i,j} \leq \gamma \dist(\vec{t}_i- j p^{i} \cdot \vec{b}_n, \lat_{i+1})$.
It then sets $\vec{t}_{i+1} = \vec{t}_{i} - j \cdot p^{i} \cdot \vec{b}_n$, where $j$ is chosen such that $d_{i,j}$ is minimal. 

Let 
\[
\vec{t}' :=  \vec{t}_\ell - p^\ell \cdot  \left\lfloor \frac{\inner{\vec{t}_\ell, \gs{\vec{b}}_n} }{p^{\ell} \cdot \| \gs{\vec{b}}_n\|^2} \right\rceil  \cdot \vec{b}_n
\; .
\] The reduction then calls itself recursively on input $\lat' := \lat(\vec{b}_1,\ldots, \vec{b}_{n-1})$ and $\vec{t}'$, receiving as output $\vec{y}' \in \lat'$. Finally, the reduction outputs $\vec{y}' + \vec{t} - \vec{t}' \in \lat$.

Take 
\[
\ell := \left\lceil  \frac{n + \log n+2}{2 \log (p/\gamma)} \right\rceil = O(n/\log p)
\; .
\] 
It is clear that the running time is as claimed. We first show by induction that $\dist(\vec{t}_i, \lat_i) \leq  \gamma^i \cdot \dist(\vec{t}, \lat)$. For $i = 0$, this is trivial. For any $i > 0$, let $\vec{x}$ be a closest vector in $\lat_{i-1}$ to $\vec{t}_{i-1}$, and assume for induction that $\| \vec{x} - \vec{t}_{i-1}\| \leq \gamma^{i-1} \dist(\vec{t}, \lat)$. Note that we can write $\vec{x}  = \vec{x}' + c p^{i-1} \cdot \vec{b}_n$, where $\vec{x}' \in \lat_{i}$ and $c \in \{ 0,\ldots, p-1 \}$. In particular, 
\[ 
d_{i, c} \leq \gamma \dist(\vec{t}_{i-1}- c p^{i-1} \cdot \vec{b}_n, \lat_i) = \gamma \| \vec{x} - \vec{t}_{i-1}\| \leq \gamma^{i} \dist(\vec{t}, \lat)
\; .
\]
It follows from the definition of $\vec{t}_i$ that $\dist(\vec{t}_i, \lat_i) \leq  d_{i, c} \leq \gamma^i \dist(\vec{t}, \lat)$, as needed.

We now wish to show that $\dist(\vec{t}', \lat') = \dist(\vec{t}_\ell, \lat_\ell)$. Suppose not. Then, clearly we have that $\dist(\vec{t}_\ell, \lat_\ell) \geq p^{\ell} \|\gs{\vec{b}}_n\|/2$. (To see this, consider the \scarequotes{distance in the direction of $\gs{\vec{b}}_n$.}) Combining this with the above, we have $\dist(\vec{t}, \lat) \geq (p/\gamma)^\ell \cdot \|\gs{\vec{b}}_n\|/2 \geq \sqrt{n} 2^{n/2}\cdot \|\gs{\vec{b}}_n\|$, contradicting Lemma~\ref{lem:babai_covering_radius}.

Combining everything together, we see that $\dist(\vec{t}', \lat') \leq \gamma^{\ell}\cdot\dist(\vec{t}, \lat)$. Finally, we assume for induction that $\|\vec{y}' - \vec{t}'\| \leq \gamma^{\ell\cdot(n-1)} \dist(\vec{t}', \vec{\lat}') \leq \gamma^{\ell \cdot n} \dist(\vec{t}, \lat)$. It follows that $\|(\vec{y}' - \vec{t}' + \vec{t}) - \vec{t}\| = \|\vec{y}' - \vec{t}'\| \leq \gamma^{\ell n} \dist(\vec{t}, \lat) = \gamma^{O(n^2/\log p)} \dist(\vec{t}, \lat)$, as needed.
\end{proof}

\subsection{The deterministic SVP reduction}

The SVP reduction uses essentially the same idea, but it is a bit more technical because the GapSVP oracle is so much weaker. The reduction is very similar to the reduction from unique SVP to GapSVP in~\cite{LM09}.

\begin{proof}[Proof of Theorem~\ref{thm:SVPdet}]
We may assume without loss of generality that $p$ is a prime and $\gamma^{3} < p$, since the result is trivial for larger $\gamma$.
On input a lattice $\lat \subset \R^n$, the reduction behaves as follows. If $n = 1$, it solves the SVP instance directly.  Otherwise, let $(\vec{b}_1,\ldots, \vec{b}_n)$ be an LLL-reduced basis for $\lat$ (which we can compute efficiently by Theorem~\ref{thm:LLL}).

Our goal is to compute a sequence of sublattices $\lat = \lat^{(0)} \subset \lat^{(1)} \subset \cdots \subset \lat^{(\ell)}$, with $\ell \geq 1$ to be set in the analysis, such that the index of $\lat^{(i+1)}$ over $\lat^{(i)}$ is $p$, and $\lambda_1(\lat^{(i+1)}) \leq \gamma \cdot \lambda_1(\lat^{(i)})$. 
In particular, we will take $\lat^{(i)} := \lat(\vec{b}_1,\ldots, \vec{b}_{n-2}, a_1^{(i)} \vec{b}_{n-1} + a_2^{(i)} \vec{b}_{n}, a_3^{(i)} \vec{b}_n)$ for some $a_{1}^{(i)}, a_{2}^{(i)}, a_{3}^{(i)}$, starting with $a_1^{(0)} := 1$, $a_2^{(0)} := 0$, and $a_3^{(0)} := 1$. To compute the remaining coefficients, the reduction behaves as follows for $i = 0,\ldots, \ell - 1$. For $j = 0,\ldots, p-1$, let $\lat_{i,j} := \lat(\vec{b}_1,\ldots, \vec{b}_{n-2}, a_1^{(i)} \vec{b}_{n-1} + a_2^{(i)} \vec{b}_{n} + j a_3^{(i)} \vec{b}_n, pa_3^{(i)} \vec{b}_n)$ and let $\lat_{i,p} := \lat(\vec{b}_1,\ldots, \vec{b}_{n-2}, p \cdot (a_1^{(i)} \vec{b}_{n-1} + a_2^{(i)} \vec{b}_{n}), a_3^{(i)} \vec{b}_n)$. For each $j$, the reduction uses its GapSVP oracle to compute $d_{i,j}$ such that $\lambda_1(\lat_{i,j}) \leq d_{i,j} \leq \gamma \lambda_1(\lat_{i,j})$. Let $j$ such that $d_{i,j}$ is minimal. The reduction then sets the coefficients so that $\lat^{(i+1)} := \lat_{i,j}$.\footnote{I.e., if $j < p$, the reduction sets $a_1^{(i+1)} := a_1^{(i)}$, $a_2^{(i+1)} := a_2^{(i)} + j a_3^{(i)}$, and $a_3^{(i+1)} := p a_3^{(i)}$. Otherwise, it sets $a_1^{(i+1)} := p a_1^{(i)}$, $a_2^{(i+1)} := a_2^{(i)}$, and $a_3^{(i+1)} := a_3^{(i)}$.}

Let $k_1$ be the largest power of $p$ that divides $a_1^{(\ell)}$, and let $k_2$ be the largest power of $p$ that divides $a_3^{(\ell)}$. If $k_1 \geq k_2$, the reduction sets $\lat' := \lat(\vec{b}_1,\ldots, \vec{b}_{n-2}, \vec{b}_{n})$ to be \scarequotes{$\lat$ without $\vec{b}_{n-1}$.} Otherwise, it sets $\lat' := \lat(\vec{b}_1,\ldots, \vec{b}_{n-2}, a_1^{(\ell)}\vec{b}_{n-1}+ a_2^{(\ell)} \vec{b}_n)$ to be \scarequotes{$\lat^{(\ell)}$ without $\vec{b}_n$.} It then calls itself recursively on $\lat'$ and returns the result.

Take 
\[
\ell := \left\lceil  \frac{n+3}{\log (p/\gamma^2)} \right\rceil = O(n/\log p)
\; .
\] 
The running time is clear. We first show that $\lambda_1(\lat^{(i+1)}) \leq \gamma \cdot \lambda_1(\lat^{(i)})$ for all $i$. Indeed, let $\vec{v} \in \lat^{(i)}$ such that $\| \vec{v} \| = \lambda_1(\lat^{(i)})$. As in the previous proof, it suffices to observe that $\vec{v} \in \lat_{i,c}$ for some $c$, as this will imply that 
\[
\lambda_1(\lat^{(i+1)}) \leq \min_i d_{i,j} \leq d_{i,c} \leq \gamma \lambda_1(\lat_{i,c}) = \gamma \|\vec{v}\| = \gamma \lambda_1(\lat^{(i)})
\; ,
\] as needed. To see this, note that we can write $\vec{v} = \sum_{i=1}^{n-2} r_i \vec{b}_i + r_{n-1}\cdot (a_1^{(i)} \vec{b}_{n-1} + a_2^{(i)} \vec{b}_{n}) + r_n \cdot a_3^{(i)} \vec{b}_n$ where $r_i \in \Z$. If $r_{n-1} \equiv 0 \bmod p$, then clearly $\vec{v} \in \lat_{i,p}$. Otherwise, there is a $c \in \{0,\ldots, p-1\}$ such that $c r_{n-1} \equiv r_n \bmod p$. Then, 
\begin{align*}
\vec{v} &= \sum_{i=1}^{n-2} r_i \vec{b}_i + r_{n-1}\cdot (a_1^{(i)} \vec{b}_{n-1} + a_2^{(i)} \vec{b}_{n}) + r_n a_3^{(i)} \vec{b}_n\\
&= \sum_{i=1}^{n-2} r_i \vec{b}_i + r_{n-1}\cdot (a_1^{(i)} \vec{b}_{n-1} + a_2^{(i)} \vec{b}_{n} + c a_3^{(i)} \vec{b}_n) + (r_n - c r_{n-1}) a_3^{(i)}\vec{b}_{n}
\; .
\end{align*}
Note that by definition $r_n - c r_{n-1} \equiv 0 \bmod p$, so it follows that $\vec{v} \in \lat_{i,c}$, as needed.

In particular, $\lambda_1(\lat^{(\ell)}) \leq \gamma^{\ell}\lambda_1(\lat)$. Now, we claim that $\lambda_1(\lat') \leq \lambda_1(\lat^{(\ell)})$. Note that any point $\vec{y} = \sum a_i \vec{b}_i \in \lat^{(\ell)} \setminus \lat'$ has either $|a_n| \geq p^{\max(k_1, k_2)} \geq p^{\ell/2}$ or $|a_{n-1}| \geq p^{\ell/2}$ (depending on whether or not $k_1 \geq k_2$). But, by Lemma~\ref{lem:LLLcoordinates}, this implies that $\|\vec{y}\| \geq 2^{-n/2 - 1} \cdot p^{\ell/2} \cdot \lambda_1(\lat) > \gamma^{\ell} \lambda_1(\lat)$. Therefore, any vector in $\lat^{(\ell)}$ of length at most $\lambda_1(\lat^{(\ell)}) \leq \gamma^\ell \cdot \lambda_1(\lat)$ must be in $\lat'$, as needed.

Finally, as in the previous proof, we can show by a simple induction argument that the output vector has length at most $\gamma^{\ell (n-1)} \lambda_1(\lat') \leq \gamma^{\ell n} \lambda_1(\lat) = \gamma^{O(n^2/\log p)} \cdot \lambda_1(\lat)$, as needed.

\end{proof}

\bibliographystyle{alpha}

\begin{thebibliography}{GMSS99}

\bibitem[AD15]{AD13}
Divesh Aggarwal and Chandan Dubey.
\newblock Improved hardness results for unique {S}hortest {V}ector {P}roblem,
  2015.
\newblock \url{http://eccc.hpi-web.de/report/2013/076/}.

\bibitem[ADRS15]{ADRS15}
Divesh Aggarwal, Daniel Dadush, Oded Regev, and Noah Stephens{-}Davidowitz.
\newblock Solving the {S}hortest {V}ector {P}roblem in $2^n$ time via discrete
  {G}aussian sampling.
\newblock In {\em {STOC}}, 2015.

\bibitem[Ajt96]{Ajt96}
Mikl{\'o}s Ajtai.
\newblock Generating hard instances of lattice problems.
\newblock In {\em {STOC}}, 1996.

\bibitem[Ajt98]{Ajtai-SVP-hard}
Mikl\'{o}s Ajtai.
\newblock The {S}hortest {V}ector {P}roblem in {L2} is {NP}-hard for randomized
  reductions.
\newblock In {\em STOC}, 1998.

\bibitem[AKS01]{AKS01}
Mikl\'{o}s Ajtai, Ravi Kumar, and D.~Sivakumar.
\newblock A sieve algorithm for the shortest lattice vector problem.
\newblock In {\em {STOC}}, 2001.

\bibitem[AR05]{AharonovR04}
Dorit Aharonov and Oded Regev.
\newblock Lattice problems in {NP} intersect {coNP}.
\newblock {\em Journal of the ACM}, 52(5):749--765, 2005.
\newblock Preliminary version in FOCS'04.

\bibitem[Bab86]{Bab86}
L.~Babai.
\newblock On {L}ov\'asz' lattice reduction and the nearest lattice point
  problem.
\newblock {\em Combinatorica}, 6(1):1--13, 1986.

\bibitem[Ban93]{banaszczyk}
W.~Banaszczyk.
\newblock New bounds in some transference theorems in the geometry of numbers.
\newblock {\em Mathematische Annalen}, 296(4):625--635, 1993.

\bibitem[BHW93]{BHW93}
U.~Betke, M.~Henk, and J.M. Wills.
\newblock Successive-minima-type inequalities.
\newblock {\em Discrete \& Computational Geometry}, 9(1):165--175, 1993.

\bibitem[BV11]{BV11}
Zvika Brakerski and Vinod Vaikuntanathan.
\newblock Efficient fully homomorphic encryption from (standard) {LWE}.
\newblock In {\em {FOCS}}. 2011.

\bibitem[BV14]{BV14}
Zvika Brakerski and Vinod Vaikuntanathan.
\newblock Lattice-based {FHE} as secure as {PKE}.
\newblock In {\em ITCS}, 2014.

\bibitem[BWS16]{BWS16}
Shi Bai, Weiqiang Wen, and Damien Stehl{\'e}.
\newblock Improved reduction from the {B}ounded {D}istance {D}ecoding problem
  to the unique {S}hortest {V}ector {P}roblem in lattices.
\newblock In {\em {ICALP}}, 2016.

\bibitem[Che13]{Cheng13}
Kuan Cheng.
\newblock Some complexity results and bit unpredictable for {S}hort {V}ector
  {P}roblem.
\newblock Cryptology ePrint Archive, Report 2013/052, 2013.
\newblock \url{http://eprint.iacr.org/2013/052}.

\bibitem[CN99]{CN99}
Jin-Yi Cai and Ajay Nerurkar.
\newblock Approximating the {SVP} to within a factor (1+1/dim$^\varepsilon$) is
  {NP}-hard under randomized reductions.
\newblock {\em Journal of Computer and System Sciences}, 59(2):221 -- 239,
  1999.

\bibitem[DK13]{DadushK13}
Daniel Dadush and Gabor Kun.
\newblock Lattice sparsification and the approximate {C}losest {V}ector
  {P}roblem.
\newblock In {\em SODA}, 2013.

\bibitem[DPV11]{DPV11}
Daniel Dadush, Chris Peikert, and Santosh Vempala.
\newblock Enumerative lattice algorithms in any norm via {M}-ellipsoid
  coverings.
\newblock In {\em {FOCS}}, 2011.

\bibitem[DRS14]{cvpp}
Daniel Dadush, Oded Regev, and Noah Stephens{-}Davidowitz.
\newblock On the {C}losest {V}ector {P}roblem with a distance guarantee.
\newblock In {\em {CCC}}, 2014.

\bibitem[Gen09]{Gen09}
Craig Gentry.
\newblock Fully homomorphic encryption using ideal lattices.
\newblock In {\em {STOC}}, 2009.

\bibitem[GMSS99]{GMSS99}
Oded Goldreich, Daniele Micciancio, Shmuel Safra, and Jean{-}Paul Seifert.
\newblock Approximating shortest lattice vectors is not harder than
  approximating closest lattice vectors.
\newblock {\em Information Processing Letters}, 71(2):55 -- 61, 1999.

\bibitem[GPV08]{GPV08}
Craig Gentry, Chris Peikert, and Vinod Vaikuntanathan.
\newblock Trapdoors for hard lattices and new cryptographic constructions.
\newblock In {\em STOC}, 2008.

\bibitem[HP14]{HP2014}
Gengran Hu and Yanbin Pan.
\newblock Improvements on reductions among different variants of {SVP} and
  {CVP}.
\newblock In {\em Information Security Applications: 14th International
  Workshop, WISA}, 2014.

\bibitem[HR12]{HRsvp}
Ishay Haviv and Oded Regev.
\newblock Tensor-based hardness of the {S}hortest {V}ector {P}roblem to within
  almost polynomial factors.
\newblock {\em Theory of Computing}, 8(23):513--531, 2012.
\newblock Preliminary version in STOC'07.

\bibitem[JS98]{JS98}
Antoine Joux and Jacques Stern.
\newblock Lattice reduction: A toolbox for the cryptanalyst.
\newblock {\em Journal of Cryptology}, 11(3):161--185, 1998.

\bibitem[Kan87]{Kan87}
Ravi Kannan.
\newblock Minkowski's convex body theorem and integer programming.
\newblock {\em Mathematics of Operations Research}, 12(3):pp. 415--440, 1987.

\bibitem[Kho05]{Khot05svp}
Subhash Khot.
\newblock Hardness of approximating the {S}hortest {V}ector {P}roblem in
  lattices.
\newblock {\em Journal of the ACM}, 52(5):789--808, September 2005.
\newblock Preliminary version in FOCS'04.

\bibitem[KL78]{KL78}
G.~A. Kabatjanski{\u\i} and V.~I. Leven{\v{s}}te{\u\i}n.
\newblock Bounds for packings on the sphere and in space.
\newblock {\em Problemy Pereda\v ci Informacii}, 14(1):3--25, 1978.

\bibitem[KS01]{KS01}
S.~Ravi Kumar and D.~Sivakumar.
\newblock On the unique shortest lattice vector problem.
\newblock {\em Theoretical Computer Science}, 255(1‚Äì2):641 -- 648, 2001.

\bibitem[Len83]{Len83}
Hendrik~W Lenstra{ }Jr.
\newblock Integer programming with a fixed number of variables.
\newblock {\em Mathematics of operations research}, 8(4):538--548, 1983.

\bibitem[LLL82]{LLL82}
A.~K. Lenstra, H.~W. Lenstra, Jr., and L.~Lov{\'a}sz.
\newblock Factoring polynomials with rational coefficients.
\newblock {\em Math. Ann.}, 261(4):515--534, 1982.

\bibitem[LLM06]{LiuLM06}
Yi-Kai Liu, Vadim Lyubashevsky, and Daniele Micciancio.
\newblock On {B}ounded {D}istance {D}ecoding for general lattices.
\newblock In {\em {RANDOM}}, 2006.

\bibitem[LM09]{LM09}
Vadim Lyubashevsky and Daniele Micciancio.
\newblock On {B}ounded {D}istance {D}ecoding, unique shortest vectors, and the
  minimum distance problem.
\newblock In {\em {CRYPTO}}, 2009.

\bibitem[LPR10]{LPR10}
Vadim Lyubashevsky, Chris Peikert, and Oded Regev.
\newblock On ideal lattices and {L}earning with {E}rrors over rings.
\newblock In {\em {EUROCRYPT}}, 2010.

\bibitem[MG02]{MicciancioBook}
Daniele Micciancio and Shafi Goldwasser.
\newblock {\em Complexity of Lattice Problems: a cryptographic perspective},
  volume 671 of {\em The Kluwer International Series in Engineering and
  Computer Science}.
\newblock Kluwer Academic Publishers, Boston, Massachusetts, March 2002.

\bibitem[Mic01]{Mic01svp}
Daniele Micciancio.
\newblock The {S}hortest {V}ector {P}roblem is {NP}-hard to approximate to
  within some constant.
\newblock {\em SIAM Journal on Computing}, 30(6):2008--2035, March 2001.
\newblock Preliminary version in FOCS 1998.

\bibitem[Mic08]{Micciancio08}
Daniele Micciancio.
\newblock Efficient reductions among lattice problems.
\newblock In {\em {SODA}}, 2008.

\bibitem[MR07]{MR07}
Daniele Micciancio and Oded Regev.
\newblock Worst-case to average-case reductions based on {G}aussian measures.
\newblock {\em SIAM Journal on Computing}, 37(1):267--302, 2007.

\bibitem[NS01]{NS01}
Phong~Q Nguyen and Jacques Stern.
\newblock The two faces of lattices in cryptology.
\newblock In {\em Cryptography and lattices}, pages 146--180. Springer, 2001.

\bibitem[Odl90]{Odl90}
Andrew~M Odlyzko.
\newblock The rise and fall of knapsack cryptosystems.
\newblock {\em Cryptology and computational number theory}, 42:75--88, 1990.

\bibitem[Pei09]{Peikert09}
Chris Peikert.
\newblock Public-key cryptosystems from the worst-case {S}hortest {V}ector
  {P}roblem.
\newblock In {\em STOC}, 2009.

\bibitem[PR07]{PR07}
Chris Peikert and Alon Rosen.
\newblock Lattices that admit logarithmic worst-case to average-case connection
  factors.
\newblock In {\em {STOC}}, 2007.

\bibitem[PS09]{PS09}
Xavier Pujol and Damien Stehl{\'e}.
\newblock Solving the shortest lattice vector problem in time $2^{2.465 n}$.
\newblock {\em IACR Cryptology ePrint Archive}, 2009:605, 2009.

\bibitem[Reg09]{Reg09}
Oded Regev.
\newblock On lattices, {L}earning with {E}rrors, random linear codes, and
  cryptography.
\newblock {\em Journal of the ACM}, 56(6):Art. 34, 40, 2009.

\bibitem[Ste15]{latticereductions}
Noah Stephens{-}Davidowitz.
\newblock Dimension-preserving reductions between lattice problems.
\newblock \url{http://noahsd.com/latticeproblems.pdf}, 2015.

\bibitem[Ste16]{DGStoSVP}
Noah Stephens{-}Davidowitz.
\newblock Discrete {G}aussian sampling reduces to {CVP} and {SVP}.
\newblock In {\em {SODA}}, 2016.

\end{thebibliography}

\end{document}